\documentclass{crm-eca}

\newcommand{\nablab}{\bar\nabla}
\newcommand{\cc}{\boldsymbol{c}}
\newcommand{\ce}{\mathcal{E}}
\newcommand{\ct}{\mathcal{T}}

\newcommand{\B}{\mathcal B}
\newcommand{\si}{\sigma}
\newcommand{\hD}{\widehat{D}}
\newcommand{\cB}{{\mathcal B}}
\newcommand{\bg}{\mbox{\boldmath{$ g$}}}
\newcommand{\II}{{\bf\rm I\hspace{-.2mm}I}}
\newcommand{\IIo}{\mathring{{\bf\rm I\hspace{-.2mm} I}}{\hspace{.2mm}}}

\newcommand{\J}{{\mbox{\sf J}}}
\newcommand{\Z}{{\mathcal Z}}

\def\sideremark#1{\ifvmode\leavevmode\fi\vadjust{\vbox to0pt{\vss
 \hbox to 0pt{\hskip\hsize\hskip1em
 \vbox{\hsize1cm\tiny\raggedright\pretolerance10000
 \noindent #1\hfill}\hss}\vbox to8pt{\vfil}\vss}}}%
                        
\usepackage[backref,pagebackref,linkcolor=blue]{hyperref}
\renewcommand{\O}{{\mathcal O}}
\usepackage{geometric-analysis-fix}  

\newtheorem{theorem}{Theorem}
\newtheorem{lemma}[theorem]{Lemma}

\newtheorem{proposition}[theorem]{Proposition}

\theoremstyle{definition}

\newtheorem{problem}[theorem]{Problem}

\theoremstyle{remark}
\newtheorem{remark}[theorem]{Remark}

\begin{document}

\papertitle[Conformal invariants via a Yamabe problem]{Generalising the Willmore equation: submanifold conformal invariants from a boundary Yamabe problem}


\paperauthor{A. Rod Gover}
\paperaddress{Department of Mathematics\\
  The University of Auckland\\
  Private Bag 92019\\
  Auckland 1142\\
  New Zealand,  and\\
  Mathematical Sciences Institute, Australian National University, ACT
  0200, Australia} 
\paperemail{r.gover@auckland.ac.nz}

\paperauthor{Andrew Waldron}
\paperaddress{Department of Mathematics\\
  University of California\\
  Davis, CA95616, USA} 
\paperemail{wally@math.ucdavis.edu}


\paperthanks{ARG supported by Marsden grant 13-UOA-018}

\makepapertitle

\begin{center}
\textbf{Abstract}
\end{center}
\begin{quote}
The Willmore energy, alias bending energy or rigid string action, and its variation---the {\it Willmore invariant}---are important surface conformal invariants with applications ranging from cell membranes to the entanglement entropy in quantum gravity. In work of Andersson, Chru\'sciel, and Friedrich, the same  invariant arises as the obstruction to smooth boundary asymptotics to the Yamabe problem of finding a metric in a conformal class with constant scalar curvature. We use
conformal geometry tools to describe and compute the asymptotics
of the Yamabe problem on a conformally compact manifold and thus produce
higher order hypersurface conformal invariants  that generalise the Willmore invariant. 
We give a holographic formula for these as well as variational principles for the lowest lying examples.
\end{quote}

\section{Introduction}

While much is known about the invariants of conformal manifolds, the same
cannot be said for the invariants of submanifolds in conformal
geometries. Codimension-1 embedded submanifolds (or {\em
  hypersurfaces}) are important  for applications
in geometric analysis and physics.
An extremely interesting example is  the Willmore equation
\begin{equation}\label{Wore}
\bar{\Delta} H +2 H(H^2-K)=0,
\end{equation}
for an embedded surface $\Sigma$ in Euclidean 3-space
$\mathbb{E}^3$. Here $H$ and $K$ are, respectively, the mean and
Gau\ss\ curvatures, while $\bar\Delta$ is the Laplacian induced on
$\Sigma$. We shall term the left hand side of this equation the {\it
  Willmore invariant}; as given, this quantity is invariant under
M\"obius transformations of the ambient $\mathbb{E}^3$. A key feature
is the linearity of its highest order term, $\bar{\Delta} H$.  This linearity
is important for PDE problems, but also means that the Willmore invariant
should be viewed as a fundamental curvature quantity.

In the 1992 article~\cite{ACF}, Andersson, Chrusciel and Friedrich (ACF)
(building on the works~\cite{AMO,AMO1,AMO2})
identified a conformal surface invariant that obstructs smooth
boundary asymptotics for a Yamabe solution on a conformally compact
3-manifold (and gave some information on the obstructions in dimension $n+1=d>3$). 
It is straightforward to show that this invariant is the
same as that arising from the variation of the Willmore energy; in
particular its specialisation to surfaces in $\mathbb{E}^3$ agrees
with (\ref{Wore}).  We show how tools from
conformal geometry can be used to describe and compute the asymptotics
of the Yamabe problem on a conformally compact manifold. This reveals
higher order hypersurface conformal invariants  that generalise the curvature
obstruction found by~ACF. In particular, for hypersurfaces of arbitrary
even dimension this yields higher order conformally invariant
analogues of the usual Willmore equation on surfaces in 3-space. The
construction also leads to a general theory for constructing and
treating conformal hypersurface invariants
along the lines of holography and  
the  Fefferman-Graham programme for constructing
invariants of a conformal structure via their Poincar\'e-Einstein and
``ambient'' metrics~\cite{FGrnew}. In this announcement, we focus only on main results,
a detailed account of this general theory will be presented elsewhere~\cite{GWnew}.

\section{The problem} \label{prob}

Given a Riemannian $d$-manifold $(M,g)$ with boundary
$\Sigma:=\partial M$, one may ask whether there is a smooth real-valued
function $u$ on $M$ satisfying the following two conditions: 
\begin{enumerate}
\item $u $ is a defining function for $\Sigma$ ({\it i.e.}, $\Sigma$
  is the zero set of $u$, and $\boldsymbol{d} u_x\neq 0$ $\forall x\in \Sigma$);
\item $\bar{g}:=u^{-2}g$ has scalar curvature
${\rm Sc}^{\bar{g}}=-d(d-1)$. 
\end{enumerate}
Here $\boldsymbol{d}$ is the exterior derivative. We assume $d\geq 3$
and all structures are $C^\infty$.

Assuming $u>0$ and setting $u=\rho^{-2/(d-2)}$, part (2) of this
problem gives the Yamabe equation.  The problem fits nicely into the
framework of conformal geometry: Recall that a conformal structure $\cc$
on a manifold is an equivalence class of metrics where the equivalence
relation $\widehat{g}\sim g$ means that $\widehat{g}= \Omega^2 g$ for
some positive function $\Omega$. The line bundle $(\Lambda^d TM)^2$ is
oriented and for $w\in \mathbb{R}$ the bundle of {\it conformal
  densities} of weight~$w$, denoted $\ce[w]$, is defined to be the
oriented $\frac{w}{2d}$-root of this (we use the same notation for
bundles as for their smooth section spaces). Locally each $g\in \cc$
determines a volume form and, squaring this, globally a section of
$(\Lambda^d T^*M)^2$. So, on a conformal manifold $(M,\cc)$ there is a
canonical section $\bg$ of $S^2T^* M\otimes \ce[2]$ called the
conformal metric. Thus each metric $g\in c$ is naturally in $1:1$
correspondence with a (strictly) positive section $\tau$ of $\ce[1]$
via $g=\tau^{-2} \bg$. Also, the Levi-Civita connection $\nabla$ of
$g$ preserves $\tau$, and hence~$\bg$.  Thus we are led to the
conformally invariant equation on a weight 1 density $\sigma\in
\ce[1]$
\begin{equation}\label{Ytwo}
S(\sigma):= 
\big(\nabla \si \big)^2 - \frac{2}{d} \si\, \Big(\Delta +\frac{\rm Sc}{2(d-1)}\Big) \si = 1 , 
\end{equation}
where $\bg$ and its inverse are used to raise and lower indices, 
$\Delta = \bg^{ab}\nabla_a\nabla_b$
and  ${\rm Sc}$ means $\bg^{bd}R_{ab}{}^a{}_d$, with $R$ the
Riemann tensor. Choosing $\cc\ni g=\tau^{-2} \bg$, equation~(\ref{Ytwo}) becomes 
exactly the PDE obeyed by the smooth function $u=\si/\tau$ solving part (2) of the problem above. 
Since $u$ is a defining function this means
$\si$ is a {\em defining density} for $\Sigma$, meaning that it is a
section of $\ce[1]$, its zero locus $\Z(\sigma)=\Sigma$, and $\nabla \si_x\neq 0$ $\forall x\in \Sigma$.  For our purpose we only
need to treat the problem formally (so it applies to any
hypersurface):

\begin{problem}\label{Riemannsfirststep} Let $\Sigma$ be an embedded
hypersurface in a conformal manifold $(M,\cc)$ with $d\geq 3$.
Given a defining density $\si$ for $\Sigma$, find a new, smooth, defining density $\bar \sigma$ such that
\begin{equation}\label{ind}
S(\bar{\si})=1 + \bar{\sigma}^{\ell} A_\ell\, ,
\end{equation}
for some $A_\ell\in \ce[-\ell]$,
where $\ell \in\mathbb{N}\cup\infty$ is as high as possible.
\end{problem}

\section{The main results} \label{results}

Here we use the notation $\O(\sigma^{\ell})$ to mean plus
$\sigma^{\ell} A$ for some smooth $A\in \ce[-\ell]$.
\begin{theorem}
\label{obstr}
Let $\Sigma$ be an oriented embedded hypersurface in 
$(M,\cc)$, where   $d\geq 3$, then: \\
$\bullet$ There
is a distinguished defining density $\bar\sigma\in \ce[1]$ for
$\Sigma$, unique to
$\O(\bar\sigma^{d+1})$, such that
\begin{equation}\label{ddens}
S({\bar \sigma})=1+\bar\sigma^d B_{\bar \sigma}\, ,
\end{equation}
where $B_{\bar \sigma} \in \ce[-d]$ is smooth on $M$. 
Given any defining density $\sigma$, then $\bar \sigma$ depends smoothly
on $(M,\cc,\sigma)$ via a canonical formula~$\bar\sigma(\sigma)$.
\\
$\bullet$ $\B:=B_{\bar\sigma(\sigma)}\big|_\Sigma$ is independent of $\sigma$
and is a natural invariant determined by $(M,\cc,\Sigma)$.
\end{theorem}

For any {\it unit conformal defining density} $\bar\si$ satisfying Eq.~(\ref{ddens}) of the
Theorem, it is straightforward, although tedious, to calculate
$\cB$. 
For $d=3$ we obtain 
\begin{equation}\label{ourWilly}
 \B=  2 \big(\bar \nabla_{(i} \bar \nabla_{j)\circ} + H\,  \IIo_{ij} +R^\top_{(ij)\circ}\big) \IIo^{ij}\, ,
 \end{equation}
where $\IIo_{ij}$ is the trace-free part of the second fundamental
form $\II_{ij}$,   $R^\top_{(ij)\circ}$ is the trace-free part of the projection of the ambient Ricci tensor 
along~$\Sigma$, and $\bar\nabla$ is the Levi-Civita for the
metric on $\Sigma$ induced by $g$.
Equation~(\ref{ourWilly})
agrees with~\cite[Theorem 1.3]{ACF} and \cite{G+Yuri} and, by using the Gau\ss--Codazzi
equations,  agrees with~(\ref{Wore}) for
$\Sigma$ in $\mathbb{E}^3$. (We note that Eq.~(\ref{ddens}) is consistent with~\cite[Lemma 2.1]{ACF}.) 

For~$d=4$ and (specializing to) conformally flat structures~$\cc$, evaluated on~$g \in \cc$ with~$g$ flat,
our result for the {\em obstruction density} $\cB$ of Theorem
\ref{obstr} is
\begin{equation}\label{4willy}
{\mathcal B}=\frac1{6}\Big((\nablab_k\IIo_{ij})\, (\nablab^k\IIo^{ij})+
2\IIo^{ij}\bar\Delta\IIo_{ij}+\frac32\, (\nablab^k\IIo_{ik})\,  (\nablab_l\IIo^{il})
-2\bar\J\, \IIo_{ij}\IIo^{ij}  
+(\IIo_{ij}\IIo^{ij})^2
\Big)\, .
\end{equation}

For $d\geq 5$ odd, we prove that the obstruction density~$\cB$ has a linear highest order term, namely
$\bar\Delta^{(d-1)/2} H$ (up to multiplication by a non-zero
constant). So: $\cB$ is an analogue of the Willmore invariant; it can
be viewed as a fundamental conformal curvature invariant for
hypersurfaces; as an obstruction it is an analogue of the
Fefferman--Graham obstruction tensor~\cite{FGrnew}.  We see this as
follows.

From the algorithm for calculating $\cB$ it is easily concluded that
it is a natural invariant (in terms of a background metric), indeed it
is given by a formula polynomially involving the second fundamental
form and its tangential (to $\Sigma$) covariant derivatives, as well
as the curvature of the ambient manifold $M$ and its covariant
derivatives.  To calculate the leading term we linearise this formula
by computing the infinitesimal variation of~$\cB$. It suffices to
consider an $\mathbb{R}$-parametrised family of embeddings of
$\mathbb{R}^{d-1}$ in $\mathbb{E}^{d}$, with corresponding defining
densities $\si_t$ and such that the zero locus $\Z(\sigma_0)$ is the
$x^{d}=0$ hyperplane (where $x^i$ are the standard coordinates on
$\mathbb{E}^{d}=\mathbb{R}^{d}$) so that $\cB|_{t=0}=0$. Then
applying $\frac{\partial}{\partial t}\mid _{t=0}$ (denoted by a dot)
we obtain the following:
\begin{proposition}\label{Bnature}
The variation of the obstruction density is given by
$$
\dot \cB = \left\{
\begin{array}{ll}
 a\cdot \bar{\Delta}^{(d+1)/2} \dot{\si} + \mbox{\em lower order terms}\, ,&d-1 \mbox{ even, with $a\neq 0$ a constant, }\\[2mm]
\mbox{\em non-linear terms}\, ,& d-1 \mbox{ odd.}
\end{array}\right.
$$
\end{proposition}
This establishes the result, as in this setting the highest order term
in the variation of mean curvature is $\frac{1}{d-1}\bar{\Delta} \dot
\si$.  It also shows that when $n$ is odd the general formula for
$\cB$ may be expressed so that it has no linear term.

Employing methods from tractor calculus~\cite{BEG}, and the notion of a holographic formula as introduced in~\cite{GW}
a simple closed formula for the obstruction density in any dimension can be obtained.
Key ingredients of this result are the tractor bundle associated to a conformal structure and the Thomas $D$-operator $D^A$. Also needed is the projector
$\Sigma^A_B:=\delta^A_B -N^A N_B$ 
 from the tractor bundle  along $\Sigma$ to the normal bundle of the normal tractor $N^A$. The latter is isomorphic to the tractor bundle of the hypersurface~$\Sigma$~\cite{Goal} and $\bar D_A$ is its intrinsic Thomas-$D$ operator. For an explanation of these details see Section~\ref{proofs} below as well as~\cite{BEG,GWnew}. In these terms our result is as follows:

\begin{theorem}
Let $\bar \sigma$ be a unit conformal defining density. Then, the ASC obstruction density $\B$
is given by the holographic formula
$$
(-1)^{n+1}\, \frac{n!(n+2)!}{4} \,  \B= \bar D_A \Big[\Sigma^A_B\Big((\bar I.D)^n \bar I^B - (\bar I.D)^{n-1}[X^B K]\Big)\Big]\Big|_\Sigma\, ,
$$
where $K:=P_{AB}P^{AB}$, $P^{AB}:=\hD^A\bar I^B$ and $\bar I^A=\hD^A \bar \sigma$.
\end{theorem}

\begin{remark}
In the above Theorem, 
the operator $(\bar I.D)^n$ is an example of a sequence of holographic formul\ae\  for tractor twistings of the conformally invariant GJMS operators of~\cite{GJMS}. These are a sequence of conformally invariant operators built from powers of the Laplacian with subleading curvature corrections; the simplest examples of these are the Yamabe operator (or conformally invariant wave operator) and Paneitz operator.
\end{remark}

The $d=3$ invariant~(\ref{ourWilly}) is the variation of the Willmore energy $E=\int_\Sigma \IIo_{ij}\IIo^{ij}$, while in $d=4$, for $\cc$ conformally flat, 
it can be shown that the invariant~(\ref{4willy}) is the variation of $E=\int_\Sigma \IIo_{ij}\IIo^{jk}\IIo_{k}{}^i$. The na\"ive conjecture that powers of traces of the trace-free second fundamental form yield integrands for  $d>4$ energy functionals
is bound to fail $d$ odd because of the leading behavior of the obstruction density given in Proposition~\ref{Bnature} (see~\cite{Guven,G+Yuri} for a study of conformally invariant $d=4$ bending energies). 
It seems likely that the higher dimensional obstruction densities are variational, therefore it is interesting to speculate whether they are the same as or closely linked to  the variations of the submanifold conformal anomalies studied in~\cite{GrW}. A related question is their relevance for entanglement entropy. Recently in~\cite{Ast} variations of the  Willmore energy were employed in a study of
surfaces maximizing entanglement entropies. 

\subsection{A holographic approach to submanifold invariants}\label{holo}
Given a conformal manifold $(M,c)$ and a section $\si$ of $\ce[1]$ one
may construct density-valued conformal invariants that couple the data
of the jets of the conformal structure with the jets of the section
$\si$.  In the setting of Theorem \ref{obstr}, consider such an
invariant $U$ (say), which uses the section $\bar\si$ of the
Theorem. Suppose that at every point, $U$ involves~$\bar\si$
non-trivially, but uses no more than its $d$-jet of $\bar\si$. Then it
follows from the first part of Theorem \ref{obstr} that $U|_\Sigma$ is
determined by $(M,c,\Sigma)$ and so is a conformal invariant of
$\Sigma$. On the interior the formula for $U$ as calculated in the
scale $\bar\si$ (so using $\bar\si$ to trivialize the density bundles)
is then a regular Riemannian invariant of $(M,\bar{g})$ (where
$\bar{g}={\bar\si}^{-2}\bg$) which corresponds holographically to the
submanifold invariant $U|_\Sigma$.  

\section{The ideas behind the main proofs}\label{proofs}
On a conformal manifold $(M,c)$, although there is no canonical
connection on $TM$, there is a canonical linear connection
$\nabla^{\ct}$ on a rank $d+2$ vector bundle known as the tractor
bundle and denoted $\ce^A$ in an abstract index notation.  A choice of
metric $g\in c$ determines an isomorphism $ \ce^A \stackrel{g}{\cong}
\ce[1]\oplus T^*\!M[1]\oplus \ce[-1] $.  This connection preserves a
metric $h_{AB}$ on $\ce^A$ that we may therefore use to raise and
lower tractor indices. For $V^A=(\si, \mu^a,
\rho)$ and $W^A=(\tau, \nu^a,\kappa)$ this is given by
$h(V,W)=h_{AB}V^A W^B=\si \kappa +\bg_{ab}\mu^a\nu^b+\rho\tau=:V.W$.
Closely linked to $\nabla^\ct$ is an important, second order
conformally invariant operator $D^A:\ce [w]\to \ce^A[w-1]$; when
$w\neq 1-\frac d2$, we denote $\frac{1}{d-2w-2}$ times this by $\hD$,
where $ \hD^A\si\stackrel{g}{=} (\si,~ \nabla_a
\si,~-\frac{1}{d}(\Delta +\J)\si )$, for the case $\si\in \ce[1]$, and
$2\J={\rm Sc}^g/(d-1)$. For $\si$ a scale, or even a defining density,
we shall write $I^A_\si:=\hD^A \si$, which we call the {\em scale
  tractor}. Now $S(\si)$ from above is just
$S(\si)=I^2_\si:= h_{AB}I^A_\si I^B_\si$,
so the equation (\ref{Ytwo}) has the nice geometric interpretation
$I_\si^2=1$~\cite{Goal}, and this is critical for our treatment.
 
Note that it is essentially trivial to solve (\ref{ind}) for the case
$\ell=1$.  Theorem \ref{obstr} is then proved inductively via the
following Lemma. The Lemma also yields an algorithm for explicit
formulae for the expansion, that we cannot explain fully here, but
through this and related results the naturality of $\cB$ can be seen.
\begin{lemma}\label{Isquare}
Suppose $I^2_\si=S(\si)$ satisfies (\ref{ind}) for $\ell=k \geq 1$.
Then, if $k\neq d$, there exists $f_k\in \ce[-k]$ such that the scale
tractor $I_{\sigma'}$ of the new defining density
$\sigma':=\sigma+\sigma^{k+1} f_k$ satisfies (\ref{ind}) for
$\ell=k+1$. 
When $k=d$ and $\sigma':=\sigma+\sigma^{d+1} f$, then for any $f\in\ce[-d]$, 
$$
I_{\sigma'}^{\, 2}=I_\sigma^{\, 2}+\O(\sigma^{d+1})\, . 
$$
\end{lemma}
\begin{proof}[Proof - Idea] First because of the scale tractor definition we have
$$ \big(\hD \sigma'\big)^2=I_\sigma^2 +\frac2d \, I_\sigma.D \big(\sigma^{k+1}
  f_k\big) + \Big[\hD \big(\sigma^{k+1} f_k\big)\Big]^2\, .
$$ Tractor calculus identities show that the last term is
  $\O(\si^{k+1})$, while $I^2_\si=1+\si^k A_k$. Crucially, the operators
  $\si$ (acting by multiplication) and $\frac{1}{I_\sigma^2}I_\si. D$
  generate an $\mathfrak{sl}(2)$ \cite{GW}. Using standard ${\mathcal
    U}\big(\mathfrak{sl}(2)\big)$ identities, we compute that $f_k:=
  -d\, A_k/(2(d-k)(k+1))$ which deals with the $k\neq d$ cases; the
  same computation gives the $k=d$ conclusion.
\end{proof}

\begin{proof}[Proof of Proposition \ref{Bnature} - sketch] 
The key idea is that for each $t$ we can replace $\si_t$ with the
corresponding normalised defining density $\bar{\si}_t$ which solves
$
I^2_{\bar{\si}_t}=1+{{\bar{\si}}_t}^{d}\mathcal{B}_{\bar{\si}_t}
$,
via Theorem \ref{obstr}, while maintaining smooth dependence on $t$.
Then it is easy to prove that
$\cB_{\bar{\si}_0}|_{\Z(\bar\si_{t=0})}=0$, while $\partial(I^2_{\bar{\si}_t})/\partial t|_{t=0}$ is  proportional to  $I. D \dot{\bar\si}$. So
applying $\frac{\partial}{\partial t}\mid _{t=0}$  implies that $\dot{\bar\si}$ solves a linear $I. D$ boundary problem  up to $\O (\bar{\si}^{d})$ with obstruction
$\dot{\cB}_{\bar{\si}}$. Using~\cite[Theorem
  4.5]{GW} we can easily deduce the conclusion.  
\end{proof}

\subsection{Acknowledgements} Prior to this work, ARG had discussions of 
this problem with 
Fernando Marques and then Pierre Albin and Rafe Mazzeo. We are
indebted for the insights so gained.


\begin{thebibliography}{99}



\bibitem{ACF} L.~Andersson, P.~Chrusciel and
  H.~Friedrich, {\em On the Regularity of solutions to the Yamabe
    equation and the existence of smooth hyperboloidal initial data
    for Einstein's field equations,}  Commun.\ Math.\ Phys.\ {\bf 149},
  587--612 (1992). 

\bibitem{Ast} A.F. Astaneh, G. Gibbons and S. N. Soludukhin, {\em What surface maximizes entanglement entropy?},
arXiv:1407.4719.

\bibitem{AMO1}
P. Aviles and R.C. McOwen, {\em Complete conformal metrics with negative scalar
curvature in compact Riemannian manifolds.} Duke Math. J., 56, 395--398 (1988).

\bibitem{AMO2} P. Aviles and R.C. McOwen, {\em Conformal deformation
  to constant negative scalar curvature on noncompact Riemannian
  manifolds.} J. Differential Geom., 27, 225--239 (1988).
 
\bibitem{BEG} T.N. Bailey, M.G. Eastwood, and A.R. Gover, {\em Thomas?s structure bundle for conformal, projective
and related structures}, Rocky Mountain J. Math. {\bf 24} (1994), 1191--1217.

\bibitem{FGrnew} C.\ Fefferman, and C.R.\ Graham, The Ambient Metric,
  Annals of Mathematics Studies, 178. Princeton University Press,
  Princeton, NJ, 2012. x+113 pp.

\bibitem{G+Yuri} Y.~Vyatkin, University of Auckland, Ph.D. thesis, 2013; A.R.~Gover and Y.~Vyatkin, in progress.

\bibitem{Goal} A.R. Gover, {\em Almost Einstein and
    Poincar\'e-Einstein manifolds in Riemannian signature},   J.\
    Geometry and Physics, {\bf 60} (2010), 182--204,   arXiv:0803.3510


\bibitem{GW} A.R.~Gover, and A.~Waldron, {\em Boundary calculus for conformally compact manifolds}, 
 Indiana U.M.J. {\bf 63} (2014) 120--163 arXiv:1104.2991. 

\bibitem{GWnew}  A.R.~Gover, and A.~Waldron, in preparation.
  
\bibitem{GJMS} C.R. Graham, R. Jenne, Ralph, L. Mason and G. Sparling, {\em Conformally invariant powers of the
Laplacian. I. Existence}, J. London Math. Soc. (2) {\bf 46} (1992), 557--565. 
  
\bibitem{GrW} C. R. Graham and E.  Witten, {\em Conformal Anomaly Of Submanifold Observables In AdS/CFT Correspondence}
Nucl. Phys. {\bf B} 546 (1999) 52--64.  
  
\bibitem{Guven} J. Guven, {\em Conformally invariant bending energy for hypersurfaces}, J. Phys. A: Math. Gen. {\bf 38} (2005) 7943--7955.  

\bibitem{AMO} C. Loewner and L. Nirenberg. {\em Partial Differential Equations Invariant under Conformal or Projective Transformations} in: contributions to Analysis, Academic Press, New York, 1974.






\end{thebibliography}
\end{document}